\documentclass[11pt]{article}
\usepackage{odonnell}

\newcommand{\remove}[1]{}

\newcommand{\hlbrt}{\mathcal{H}} \newcommand{\bounded}{\mathcal{B}(\hlbrt)}   

\newcommand{\eigenspace}[1]{V_1({#1})}
\newcommand{\eigenspaceprime}{\eigenspace{\wh{\Phi}}}
\newcommand{\dimfixed}{m}

\newcommand{\gnote}[1]{}
\newcommand{\prm}{p} 

\begin{document}

\title{Quantum automata cannot detect biased coins, even in the limit}
\author{Ryan O'Donnell\thanks{Computer Science Dept., Carnegie Mellon Univ.  Supported by NSF grant CCF-1618679. \texttt{odonnell@cs.cmu.edu}}
\and Guy Kindler\thanks{School of Computer Science and Engineering, Hebrew
Univ.\ of Jerusalem.  Supported by BSF grant 2012220. \texttt{gkindler@cs.huji.ac.il}}}

\maketitle

\begin{abstract}
    Aaronson and Drucker (2011) asked whether there exists a quantum finite automaton that can distinguish fair coin tosses from biased ones by spending significantly more time in accepting states, on average, given an infinite sequence of tosses.  We answer this question negatively.
\end{abstract}

\section{Introduction}              \label{sec:intro}

In a 2011 work, Aaronson and Drucker~\cite{AD11} investigated the ability of a finite automaton to distinguish, given an infinite sequence of coin tosses, whether the coins are fair or $(\frac12 \pm \eps)$-biased.  There are several axes of consideration discussed in~\cite{AD11}, three of which we state here:
\begin{enumerate}
    \item Whether the automaton is classical (and probabilistic), or quantum.
    \item Whether $\eps > 0$ is ``known'' or not; i.e., whether the automaton can depend on~$\eps$.
    \item The mechanism by which the automaton makes its decision.  One possibility is that the automaton guesses ``biased'' by halting, and guesses ``fair'' by running forever.  A laxer possibility is that the automaton always runs forever, with each of its states designated ``biased'' or ``fair''; its final decision is based on the limiting time-average it spends in ``biased'' vs.\ ``fair'' states.  We refer to the two mechanisms as ``one-sided halting'' and ``limiting acceptance''.
\end{enumerate}
For example, an old result of Hellman and Cover~\cite{HC70} is that
even when $\eps$ is known and limiting acceptance is allowed, a
classical automaton needs~$\Omega(1/\eps)$ states to solve the
problem.  On the other hand, Aaronson and Drucker made the interesting
observation that for every fixed known~$\eps$, there's a quantum
automaton with just~$2$ states that solves the problem using one-sided
halting.  They also showed no
quantum automaton with a fixed number of states can solve the problem
for every \emph{unknown}~$\eps$, if the decision mechanism is
one-sided halting.

Aaronson and Drucker asked whether
the same negative result holds even if the automaton is allowed to use
the limiting acceptance decision mechanism.  Indeed, for the 48
different variations of the problem they considered, this was the only
variant that remained unsolved.  In 2014, Aaronson called this
question one of the ``Ten Most Annoying Problems in Quantum
Computing''~\cite{Aar14}.

In this work, we make the world of quantum computing $10\%$ less
annoying by resolving the problem in the negative.  Stated
informally, our main theorem
is the following (a precise phrasing appears below after we give some
formal definitions):
\begin{theorem} \label{thm:main} There is no quantum finite automaton
  that has the following property, simultaneously for every
  $\eps \in[-\frac12,\frac12]\setminus\{0\}$: Given access to an infinite sequence of coin tosses,
  if the coin is $(\half + \eps)$-biased then the automaton spends
  at least~$2/3$ of its time guessing ``biased'', and if the coin is
  fair then the automaton spends at least~$2/3$ of its time guessing
  ``fair''.
\end{theorem}

\noindent Proving this theorem involves a careful understanding of the fixed points of quantum channels.

\section{Classical and quantum automata}
In this section we review the definitions of probabilistic and quantum finite state automata.  Although we are ultimately only concerned with quantum automata, we feel it is instructive to also discuss probabilistic automata at the same time. All of our automata will have input alphabet $\Sigma = \{0,1\}$, which may be thought of as $\{\text{tails}, \text{heads}\}$.

A classical deterministic automaton on alphabet $\Sigma = \{0,1\}$ has some $d$ \emph{basic-states},\footnote{There is an unfortunate terminology clash involving the word ``state'' --- in automata theory, ``states'' are the basic vertices in automaton graphs, whereas in quantum mechanics a ``state'' usually means the ``mixed quantum state'' or ``density operator'' of a given system. Throughout we'll refer to the former as ``basic-states'' and the latter as ``quantum-states''.} an initial basic-state $i_0 \in [d]$, and transition rules $f_0, f_1 : [d] \to [d]$. Given a sequence of input symbols $w_1, w_2, w_3, \dots \in \{0,1\}$, the automaton operates as follows:  It starts in basic-state $i_0$ at time~$0$. Then, if it is in basic-state $i_{t}$ at time~$t \in \N$, it transitions to basic-state $f_{w_{t+1}}(i_{t})$ at time~$t+1$.  Automata also typically have their basic-states classified as ``accept'' or ``reject''; we discuss this more later.

One can also consider classical \emph{probabilistic} automata. These have randomized transitions, which can be encoded by a pair of $d \times d$ \emph{stochastic matrices} $S_0, S_1$.  Now at any time~$t$ the automaton can be in a ``probabilistic-state'', represented by a length-$d$ probability vector $\pi_t$.  (An initial probabilistic-state~$\pi_0$ is also specified.)  On reading symbol~$w_{t+1}$, the automaton transitions to the probabilistic-state $\pi_{t+1} = S_{w_{t+1}} \pi_t$.

Finally, the setting for a \emph{quantum} automaton is a $d$-dimensional Hilbert space~$\hlbrt$ (which we may think of as having an orthonormal basis of ``basic-state vectors'' $\ket{1}, \dots, \ket{d}$).  At any time~$t$, the automaton has a ``quantum-state'', which is a density operator $\rho_t \in \bounded$.  Here $\bounded$ denotes the set of linear operators on~$\hlbrt$, and a density operator means a positive semidefinite operator of trace~$1$.  (Probabilistic-states are the special case of quantum-states in which $\rho_t$ is diagonal with respect to $\ket{1}, \dots, \ket{d}$.)  The transition rules are now any two allowable quantum transformations~$\Phi_0, \Phi_1$; i.e., they are \emph{quantum channels (superoperators)} on~$\bounded$.  Here a quantum channel means a linear map $\Phi : \bounded \to \bounded$ that is completely positive and trace-preserving; an equivalent condition is that there exist (non-unique) \emph{Kraus operators} $K_1, \dots, K_r \in \bounded$ with $\sum_{i=1}^r K_i^\dagger K_i = \bbone$ such that $\Phi(\rho) = \sum_{i=1}^r K_i \rho K_i^\dagger$.  (For more on quantum channels, see e.g.~\cite{Wol12}.)  Again, an initial quantum-state~$\rho_0$ is given, and on reading symbol~$w_{t+1}$, the automaton transitions from quantum-state $\rho_t$ to quantum-state $\rho_{t+1} = \Phi_{w_{t+1}}(\rho_t)$.

\paragraph{Automata with random inputs.} This paper is concerned with automata whose inputs are infinite sequences of $\prm$-biased coin tosses, $p \in [0,1]$.  More formally, we always assume the input symbols $w_1, w_2, w_3, \dots \in \{0,1\}$ are chosen independently at random with $\Pr[w_t = 1] = \prm$.  Because of this assumption, we can give a simplified formalization of probabilistic and quantum automata.  In the case of probabilistic automata, at each time step (independently) we apply~$S_1$ with probability~$\prm$ and $S_0$ with probability~$1-\prm$.  It is clear that this is equivalent to simply applying the stochastic matrix $S_\prm \coloneqq \prm S_1 + (1-\prm)S_0$ at each time step.  In other words, the probabilistic-state of a probabilistic automaton after~$t$ time steps is simply $S_\prm^t \pi_0$.  The setup is precisely equivalent to a Markov chain on~$[d]$ with transition matrix~$S_\prm$.

Similarly for quantum automata, at each time step we apply~$\Phi_1$ with probability~$p$ and $\Phi_0$ with probability~$1-\prm$; this is physically equivalent to simply applying the channel $\Phi_\prm \coloneqq \prm \Phi_1 + (1-\prm)\Phi_0$ at each time step.  (This is ultimately because being in quantum-state $\rho$ with probability~$\prm$ and quantum-state $\rho'$ with probability~$1-\prm$ is physically equivalent to being in quantum-state $\prm\rho + (1-\prm)\rho'$.) Thus the quantum-state of a probabilistic automaton after~$t$ time steps is simply~$\Phi_\prm^t(\rho_0)$; we have here the quantum analogue of a Markov chain.

\paragraph{Automaton acceptance probability.}  As discussed in Section~\ref{sec:intro}, we will be considering ``limiting acceptance'', the most relaxed possible notion for automaton acceptance.  We first define this in the context of probabilistic automata.  Here, each basic-state in~$[d]$ is classified as either guessing ``Fair'' or ``Biased''.  We write $e_{\text{fair}} \in \R^d$ for the $0$-$1$ indicator of the Fair states. 
Thus if the automaton is in probabilistic-state $\pi \in \R^d$, the probability it is in a Fair basic-state is $\langle e_{\text{fair}}, \pi\rangle$.  We then consider, for a sequence of~$T$ coin tosses, the \emph{average} probability with which the automaton is in a Fair basic-state:
\[
    f_T(\prm) \coloneqq \frac{1}{T} \sum_{t=1}^T \langle e_{\text{fair}}, S_\prm^t \pi_0 \rangle = \Bigl\la e_{\text{fair}}, \Bigl(\frac{1}{T} \sum_{t=1}^T  S_\prm^t\Bigr) \pi_0 \Bigr\rangle.
\]
Finally, we consider the limiting value of this probability:
\[
    f(\prm) \coloneqq \lim_{T \to \infty} f_T(\prm) = \la e_{\text{fair}}, S_\prm^\infty \pi_0 \ra, \quad \text{where } S_\prm^\infty \coloneqq \lim_{T \to \infty} \frac{1}{T} \sum_{t=1}^T  S_\prm^t.
\]
Here we relied on the well-known fact that the limiting matrix $S_\prm^\infty$ exists.  (In fact, $S_\prm^\infty$ is also a stochastic matrix, and it acts by projection onto the $1$-eigenspace of $S_\prm$; we discuss this further in Section~\ref{sec:outline}.)  One may then say that the probabilistic automaton ``guesses Fair in the limit'' if $f(\prm) \geq \frac23$, and ``guesses Biased in the limit'' if $f(\prm) \leq \frac13$.  (It may be considered ``indecisive'' otherwise.)

The definitions for a quantum automaton are extremely similar.  The automaton is assumed to come equipped with an ``acceptance POVM'', $\{E_{\text{fair}}, \bbone - E_{\text{fair}}\}$. (Here $E_{\text{fair}} \in \bounded$ is any operator satisfying $0 \preceq E_{\text{fair}} \preceq \bbone$, and $\bbone$ denotes the identity operator.)  If the automaton is in quantum-state~$\rho$, the probability of it measuring ``Fair'' is $\la E_{\text{fair}}, \rho\ra \coloneqq \tr(E_{\text{fair}}^\dagger\rho)$.  We can then again define the limiting average probability of guessing ``Fair'' via
\begin{align}
    f_T(\prm) &\coloneqq \frac{1}{T} \sum_{t=1}^T \langle E_{\text{fair}}, \Phi_\prm^t \pi_0 \rangle = \Bigl\la E_{\text{fair}}, \Bigl(\frac{1}{T} \sum_{t=1}^T  \Phi_\prm^t\Bigr) \pi_0 \Bigr\rangle, \nonumber \\
    f(\prm) &\coloneqq \lim_{T \to \infty} f_T(\prm) = \la E_{\text{fair}}, \Phi_\prm^\infty \pi_0 \ra, \quad \text{where } \Phi_\prm^\infty \coloneqq \lim_{T \to \infty} \frac{1}{T} \sum_{t=1}^T  \Phi_\prm^t. \label{eqn:f}
\end{align}
Again, it is known that the limiting operator $\Phi_\prm^\infty$ exists; this is explicitly discussed in Section~\ref{sec:outline}.  As before, one may say that the quantum automaton ``guesses Fair in the limit'' if $f(\prm) \geq \frac23$, and ``guesses Biased in the limit'' if $f(\prm) \leq \frac13$.

We may now state the main theorem of this paper:
\begin{theorem}                                     \label{thm:it's-continuous}
    In the setting of quantum automata reading $\prm$-biased bits (as described above), the function~$f$ from~\eqref{eqn:f} is a continuous function of $\prm \in (0,1)$.
\end{theorem}
This theorem is a formal strengthening of Theorem~\ref{thm:main}, our negative result for coin distinguishing stated in Section~\ref{sec:intro}. For example, it implies that if an automaton guesses ``Fair'' in the limit'' for $\prm = \frac12$, then for all sufficiently small~$\eps$ it \emph{cannot} guess ``Biased'' in the limit for $\prm = \frac12 \pm \eps$. In fact, we get the inability of quantum automata to distinguish $\prm$-biased and $(\prm \pm \eps)$-biased coins with limiting acceptance for any fixed $\prm \in (0,1)$.  As noted in~\cite{AD11}, this is sharp in the sense that there is a trivial $2$-state deterministic classical automaton that distinguishes a $0$-biased coin from any $\eps$-biased coin, even with one-sided halting.

\section{Outline of the proof} \label{sec:outline}
Here we give an outline of the proof of Theorem~\ref{thm:it's-continuous}.  At the same time, it will be instructive to outline the analogous proof in the special case of probabilistic automata.  To prove that the limiting acceptance probability~$f(\prm)$ from~\eqref{eqn:f} is continuous for $\prm \in (0,1)$, it is enough to prove the following:
\begin{theorem} \label{thm:operator-cts}
     $\Phi^\infty_\prm$ is continuous for~$\prm \in (0,1)$.
\end{theorem}
\noindent Here for definiteness we can take the metric on channels induced by the operator norm on~$\bounded$; Theorem~\ref{thm:it's-continuous} then follows because matrix multiplication and inner product are continuous.\\

Now is a good time to review the properties of~$\Phi^\infty_\prm$.  In general, let $\Phi$ denote any channel on~$\bounded$.  Then the following are known~\cite[Prop.~6.3]{Wol12} (and easy) facts: First,
$
    \Phi^\infty \coloneqq \lim_{T \to \infty} \frac{1}{T} \sum_{t=1} ^T \Phi^t
$
exists and is itself a channel.  Second, as an operator $\Phi^\infty$, acts as projection onto the fixed points~$\eigenspace{\Phi}$ of~$\Phi$. Here we are using the following notation:
\begin{notation}
    For any operator~$A$ we write $\eigenspace{A}$ for the eigenspace of~$A$ with eigenvalue~$1$, i.e., the \emph{invariant} subspace for~$A$.
\end{notation}
\noindent As mentioned earlier, the analogous statements are true regarding~$S^\infty$, when $S$ is a stochastic operator.  (In both the probabilistic and quantum cases, the essential point is that the operator in question has spectral radius~$1$.)\\

Returning to Theorem~\ref{thm:operator-cts}, certainly $\Phi_\prm = \prm \Phi_1 + (1-\prm) \Phi_0$ varies continuously for~$\prm \in [0,1]$. But what we need to prove is that the \emph{invariant subspace} $\eigenspace{\Phi_\prm}$ of $\Phi_\prm$ varies continuously for $\prm \in (0,1)$.  There is one obvious potential obstruction: the \emph{dimension} of~$\eigenspace{\Phi_\prm}$ might change as~$\prm$ varies.  (As we will see, this is actually the only obstruction.)  Now in general, slightly perturbing a matrix \emph{can} change the dimension of its $1$-eigenspace.  However we are not concerned with completely general perturbations: we are just considering all the convex combinations of two fixed channels~$\Phi_0, \Phi_1$. The main technical theorem in our paper will be the following:
\begin{theorem} \label{thm:same-dim2} For any channels $\Phi_0, \Phi_1$, the dimension $\dim \eigenspace{\Phi_\prm}$ is the same for all $\prm \in (0,1)$.
\end{theorem}
We will discuss the intuition for this theorem below.  But first we will observe that Theorem~\ref{thm:operator-cts} is an elementary linear-algebraic consequence of Theorem~\ref{thm:same-dim2}. This deduction of Theorem~\ref{thm:operator-cts} from Theorem~\ref{thm:same-dim2} is a little more familiar if we consider $\bbone - \Phi_\prm$ rather than~$\Phi_\prm$.  Then $\Phi_\prm^\infty$ is the projection onto the kernel of $\bbone - \Phi_\prm$, and it is elementary that, given a continuously-parameterized family of matrices like $p \mapsto \bbone - \Phi_\prm$, the kernel varies continuously wherever the nullity (in this case, $\dim \eigenspace{\Phi_\prm}$) is locally constant. For a simple explicit proof see, e.g.,~\cite{use15}.\\

Thus all that remains in this work is to prove Theorem~\ref{thm:same-dim2}.  We will do this in Section~\ref{sec:technical-proof}, but first we provide some intuition and introduce a key definition, that of \emph{combinatorially equivalent channels}.

\subsection{Intuition for Theorem~\ref{thm:same-dim2}}
All of our discussion so far applies equally to probabilistic automata defined by stochastic matrices~$S_0, S_1$.  So let us first consider the analogue of Theorem~\ref{thm:same-dim2} in this case.  Here we have a family of Markov chains defined by $S_\prm = \prm S_1 + (1-\prm) S_0$ and we want to consider the dimension of their invariant subspaces.  It is well known that the invariant subspace~$\eigenspace{S}$ of the Markov chain defined by~$S$ is spanned by a linearly independent set of invariant \emph{probabilistic-states}.  Thus $\dim \eigenspace{S}$ is equal to the number of linearly independent (``fundamentally different'', one might say) invariant distributions.

In the study of Markov chains, it's popular to focus on the irreducible case, in which case there is a unique invariant probability distribution; i.e., $\dim \eigenspace{S} = 1$.  However in general we must consider reducible Markov chains (the ``mathematically annoying case'', as Hellman and Cover~\cite{HC70} put it).  Fortunately, the theory of reducible Markov chains is well developed, and it is known that there is one linearly independent invariant distribution per every \emph{communication class} of the Markov chain.  Here the ``communication classes'' of the Markov chain defined by~$S$ are precisely the strongly connected components of the underlying digraph on~$[d]$; i.e., the graph which has a directed edge $(i,j)$ whenever $S_{ij} \neq 0$.  Given this theory, it is easy to deduce the analogue of Theorem~\ref{thm:same-dim2}; the point is that for any fixed $S_0, S_1$, \emph{the underlying digraph of~$S_\prm$ is the same for all $\prm \in (0,1)$}.  Since $S_\prm = \prm S_1 + (1-\prm)S_0$, an edge $(i,j)$ is present in $S_\prm$ is present if and only if it is present in both~$S_0$ and~$S_1$.  Thus $S_\prm$ has the same set (hence number) of communication classes for all~$\prm \in (0,1)$, as needed.

In this paper, we show there is an analogous sequence of ideas in the quantum case, using some of the recently developed theory of fixed points of quantum channels.  Given a quantum channel~$\Phi$, it is known~\cite[Cor.~6.5]{Wol12} that $\eigenspace{\Phi}$ is always spanned by linearly independent \emph{quantum-states}. The analogous notion to communication classes is that of \emph{minimal enclosures}.  Further, similar to how the communication classes of a Markov chain are determined only by the nonzero pattern of its transition matrix, the minimal enclosures of a quantum channel are determined only by its Kraus operators.  We then make use of the fact that all the convex combinations $\Phi_\prm$ of two channels $\Phi_0, \Phi_1$ have related Kraus operators.  Specifically, we introduce the following notion:
\begin{definition} 	
    We will say that two channels $\Phi$ and $\wh{\Phi}$ (with the same Hilbert space~$\hlbrt$) are \emph{combinatorially equivalent} if there are Kraus operators $K_1, \dots, K_r$ for $\Phi$ and $\wh{K}_1, \dots, \wh{K}_{\wh{r}}$ for $\wh{\Phi}$ such that each $K_i$ is proportional to some $\wh{K}_{i'}$ and vice versa.
\end{definition}
Given channels $\Phi_0, \Phi_1$ with Kraus operators $\{K^{(0)}_i : i \in [r_0]\}, \{K^{(1)}_j : j \in [r_1]\}$ respectively, the channel $\Phi_\prm = \prm \Phi_1 + (1-\prm) \Phi_0$ has Kraus operators $\{\sqrt{1-\prm} K^{(0)}_i : i \in [r_0]\} \cup \{\sqrt{\prm} K^{(1)}_j : j \in [r_1]\}$.  Thus the channels $\Phi_\prm$ are all pairwise combinatorially equivalent for $\prm \in (0,1)$ (though not necessarily for $\prm \in \{0,1\}$). To show Theorem~\ref{thm:same-dim2}, it therefore suffices to show the following more general result:
\begin{theorem}	\label{thm:same-dim} 	
    Suppose $\Phi$ and $\wh{\Phi}$ are combinatorially equivalent.  Then $\dim \eigenspace{\Phi} = \dim \eigenspace{\wh{\Phi}}$.
\end{theorem}

\section{The last step: proof of Theorem~\ref{thm:same-dim}}\label{sec:technical-proof}

To prove Theorem~\ref{thm:same-dim}, we use some known results concerning the decomposition of a quantum channel into irreducible components, and the structure of its invariant quantum-states.  We will specifically use the key decomposition theorem appearing variously as \cite[Theorem~6.14]{Wol12}, \cite[Theorem~7]{BN12}, \cite[Theorem~7.2]{CP15}.

Let $\Phi$ denote a quantum channel on $\bounded$ with Kraus operators $K_1, \dots, K_r$.  We are interested in  $\dimfixed = \dim \eigenspace{\Phi}$, the dimension of the space of $\Phi$'s fixed points. As $\Phi$ is a quantum channel, it is known~\cite[Prop.~6.1]{Wol12} that its spectral radius is~$1$ and that it has at least one eigenvalue equal to~$1$; thus $\dimfixed \geq 1$.  As mentioned, it is also known~\cite[Cor.~6.5]{Wol12} that $\eigenspace{\Phi}$ is always spanned by some~$\dimfixed$ linearly independent \emph{quantum-states}.

If $\rho$ is a quantum-state, its \emph{support} $\supp(\rho)$ is simply the range of~$\rho$ as a subspace of $\hlbrt$.
The \emph{recurrent subspace} for $\Phi$ is the subspace of $\hlbrt$ defined by
\[ 	
    \calR = \spn\{\supp(\rho) : \text{$\rho$ is an invariant quantum-state}\}.
\]
The orthogonal complement of $\calR$ in $\calH$ is denoted $\calD$; this is the \emph{decaying} (or \emph{transient}) subspace.  A subspace $\calV \subseteq \hlbrt$ is called an \emph{enclosure} if $\supp(\rho) \subseteq \calV \implies \supp(\Phi(\rho)) \subseteq \calV$ for all quantum-states~$\rho$.  We can relate this concept to Kraus operators via the following equivalence:
\begin{fact} \label{fact:kraus-encl} 	
    (\cite[Proposition~4.4]{CP15}.) $\calV$ is an enclosure if and only if $K_i \calV \subseteq \calV$ for all Kraus operators~$K_i$. 	
\end{fact}
An enclosure $\calV$ is called \emph{minimal} if it is nonzero and all enclosures $\calV' \subseteq \calV$ are equal to either $\{0\}$ or~$V$.  It is also known~\cite[Prop.~15]{BN12} that a subspace of~$\hlbrt$ is a minimal enclosure if and only if it is the support of an extremal invariant quantum-state, meaning one that cannot be written as a nontrivial convex combination of two distinct invariant quantum-states.  One consequence is that
\begin{equation}	\label{eqn:decay} 	
    \calR = \spn\{\supp(\rho) : \text{$\rho$ is an extremal invariant quantum-state}\} = \spn\{\calV : \text{$\calV$ is a minimal enclosure}\}.	
\end{equation}

The theorems \cite[Theorem~6.14]{Wol12}, \cite[Theorem~7]{BN12}, \cite[Theorem~7.2]{CP15} characterize~$\eigenspace{\Phi}$ and the quantum-states therein in slightly different ways.  To explain, we make some
definitions.
\begin{definition}	\label{def:min-encl-decomp}
    (In this definition, $k$, $m_1, \dots, m_k$, $d_1, \dots, d_k$ denote positive integers.)

	Given $\Phi$, we define a \emph{minimal enclosure decomposition} to be an orthogonal decomposition of
$\hlbrt$ into subspaces
    \begin{equation} \label{eqn:decomp}
		\hlbrt = \calD \oplus \bigoplus_{i=1}^{k} \calW_{i}, \quad \text{where }\calW_i = \bigoplus_{j=1}^{m_{k}} \calV_{i,j}
    \end{equation}
	in which $\calD$ is the decaying subspace for~$\Phi$, each $\calV_{i,j}$ is a minimal enclosure, $\dim
\calV_{i,j} = d_i$ for all $1 \leq j \leq m_i$, and finally the following property holds: \emph{For any minimal
enclosure $\calX$ of~$\Phi$ and any $1 \leq i \leq k$, if $\calX$ is not orthogonal to~$\calW_i$ then $\calX
\subseteq \calW_i$.  (In particular, if $m_i = 1$ then $\calX$ must equal~$\calW_i$.)}
\end{definition}
\begin{definition} 	
    Suppose we have a minimal enclosure decomposition for~$\Phi$ as above.  Fix any ordered
orthogonal basis for $\hlbrt$ compatible with~\eqref{eqn:decomp} (meaning the first $\dim \calD$ elements span
$\calD$, the next $m_1 d_1$ elements come in $m_1$ groups of $d_1$ spanning $\calV_{1,1}, \dots, \calV_{1,m_1}$
respectively, etc.).  Let $X \in \bounded$, and think of~$X$ in its matrix form with respect to the ordered
basis.

    Then we say that $X$ \emph{respects the minimal enclosure decomposition} if $X$ is block-diagonal with blocks
    corresponding to $\calD$, $\calW_1, \dots, \calW_k$, and furthermore $X$ is~$0$ on the $\calD$-block and is
    of the form $A_i \otimes \rho_i$ on the $\calW_i$-block for some $A_i \in \calC^{m_i \times m_i}$ and some
    strictly positive density matrix $\rho_i \in \calC^{d_i \times d_i}$. In symbols,
    \[
    	X = 0 \oplus \bigoplus_{i=1}^k A_i \otimes \rho_i.
    \]
	(We remark that the property of respecting the minimal enclosure decomposition does not depend on the choice
of the compatible orthogonal basis.)
\end{definition}

In combination, \cite[Theorem~6.14]{Wol12}, \cite[Theorem~7]{BN12} state the following:\footnote{\cite{Wol12} deals with the invariant subspace whereas~\cite{BN12} deals with the invariant quantum-states.  The fact that the $\rho_i$'s are strictly positive is in \cite{Wol12}. Finally, \cite{BN12} does not explicitly show that the minimal enclosure decomposition satisfies the last, italicized, condition in Definition~\ref{def:min-encl-decomp}. However, it's implicit and it's easy to deduce: we know that any minimal enclosure $\calX$ is the support of some extremal invariant quantum-state~$\rho$, and it's clear that if this support is not entirely within a single $\calW_i$-block then $\rho$ would not be extremal.}
\begin{theorem}	\label{thm:BN} 	
    Given any channel~$\Phi$, there exists a minimal enclosure
    decomposition as in~\eqref{eqn:decomp} such that
    $\eigenspace{\Phi}$ consists precisely of \emph{all}~$X \in
    \bounded$ that respect the decomposition. (An immediate
    consequence is that $\dimfixed = \dim \eigenspace{\Phi} = \sum_i
    m_i^2$.)  Finally, the quantum-states that are invariant are
    precisely  all such~$X$ with $A_i = \lambda_i \sigma_i$, where $\sigma_1, \dots, \sigma_k$ are density matrices and $\lambda_1, \dots, \lambda_k$ are nonnegative reals summing to~$1$.
\end{theorem}
The statement of \cite[Theorem~7.2]{CP15} is slightly different:\footnote{The first statement of this theorem is \cite[Proposition~7.1]{CP15}, except that that Proposition does not include the final, italicized, property of Definition~\ref{def:min-encl-decomp} for those~$i$ with $m_i = 1$.  However it is evident from the proof that this is an oversight; a person communication from the authors confirmed this.}
\begin{theorem}	\label{thm:CP} 	Given any channel~$\Phi$, at least one minimal enclosure decomposition exists.  Furthermore, given \emph{any} minimal enclosure decomposition 	
    \[ 		
        \hlbrt = \calD \oplus \bigoplus_{i=1}^{\wh{k}} \wh{\calW}_{i}, \quad \text{where }\wh{\calW}_i =
\bigoplus_{j=1}^{\wh{m}_{\wh{k}}} \wh{\calV}_{i,j}, 	
    \]
    every invariant quantum-state for~$\Phi$ respects it.  (As an immediate consequence, ${\dimfixed = \dim
    \eigenspace{\Phi} \leq \sum_i \wh{m}_i^2}$.)
\end{theorem}

We are now able to give the proof of Theorem~\ref{thm:same-dim}.

\begin{proof}[Proof of Theorem~\ref{thm:same-dim}.]
    Write $\dimfixed = \dim \eigenspace{\Phi}$ and $\wh{\dimfixed} = \dim \eigenspaceprime$.  Since $\Phi$ and $\wh{\Phi}$ play symmetric roles, it suffices to show $\wh{\dimfixed} \leq \dimfixed$.  Apply Theorem~\ref{thm:BN} to $\Phi$, obtaining a minimal enclosure decomposition as in~\eqref{eqn:decomp}.  We have $\dimfixed = \sum_{i=1}^{k} m_i^2$.  We claim that this decomposition is also a minimal enclosure decomposition for~$\wh{\Phi}$.  This will finish the proof of $\wh{\dimfixed} \leq \dimfixed$, by Theorem~\ref{thm:CP}.

    To see the claim, we first observe that every enclosure~$\calV$ for $\Phi$ is an enclosure for $\wh{\Phi}$
    (and vice versa).  This follows from Fact~\ref{fact:kraus-encl}: $\calV$~satisfies $K_i \calV \subseteq
    \calV$ for each Kraus operator $K_i$ of~$\Phi$, and hence the same is true for the Kraus operators
    $\wh{K}_{i'}$ of~$\wh{\Phi}$, by combinatorial equivalence of $\Phi$ and $\wh{\Phi}$.  It then follows by definition that every \emph{minimal} enclosure for $\Phi$ is also a minimal enclosure for~$\wh{\Phi}$ (and vice versa).  Finally, the claim now follows because $\Phi$ and $\wh{\Phi}$ have the same decaying subspace (by~\eqref{eqn:decay}) and because Definition~\ref{def:min-encl-decomp} of minimal enclosure decompositions depends \emph{only} on which subspaces of~$\hlbrt$ are minimal enclosures.
\end{proof}

\bibliographystyle{alpha}
\bibliography{odonnell-bib}
\end{document}